\documentclass{article}
\usepackage{geometry}
\usepackage{verbatim}
\usepackage{float}
\usepackage{booktabs}
\usepackage{amsthm}
\usepackage{amsmath}
\usepackage{amssymb}
\usepackage{bbm}
\usepackage{graphicx}
\usepackage{setspace}
\usepackage{xcolor}
\usepackage[most]{tcolorbox}
\usepackage[numbers]{natbib}
\usepackage{hyperref}
\usepackage{enumitem}
\usepackage{bbm}
\usepackage{cleveref}
\usepackage{multirow}
\usepackage{subcaption}
\usepackage{algorithm}
\usepackage{algorithmic}
\usepackage[utf8]{inputenc} 
\usepackage[T1]{fontenc}
\usepackage{adjustbox}
\usepackage{makecell}  
\usepackage{pdflscape}
\usepackage{longtable,booktabs,array}
\usepackage{ragged2e}
\usepackage{fvextra}

\usepackage[english]{babel}

\newtheorem{proposition}{Proposition}

\usepackage{tabularx}     

\title{Eligibility-Aware Evidence Synthesis: An Agentic Framework for Clinical Trial Meta-Analysis}
   \author{\small
    Yao Zhao$^{1,*}$, \quad Zhiyue Zhang$^{1,*}$, \quad Yanxun Xu$^{1, \dagger}$\\
   \small $^1$ Department of Applied Mathematics and Statistics, Johns Hopkins University \\
   \small  $^*$ These authors contributed equally to this work\\
    \small  $\dagger$ Correspondence should be addressed to \texttt{yanxun.xu@jhu.edu}} 
    \date{}

\begin{document}
\maketitle

\begin{abstract}
Clinical evidence synthesis requires identifying relevant trials from large registries and aggregating results that account for population differences. While recent LLM-based approaches have automated components of systematic review, they remain fragmented and do not support end-to-end evidence synthesis. Moreover, conventional meta-analysis weights studies by statistical precision without accounting for clinical compatibility reflected in eligibility criteria. We propose EligMeta, an agentic framework that integrates automated trial discovery with eligibility-aware meta-analysis, translating natural-language queries into reproducible trial selection and incorporating eligibility alignment into study weighting to produce cohort-specific pooled estimates. EligMeta employs a hybrid architecture separating LLM-based reasoning from deterministic execution: LLMs generate interpretable selection rules from natural-language queries and perform schema-constrained parsing of trial metadata, while all logical operations, comparisons, eligibility-weight computations, and statistical pooling are executed deterministically to ensure reproducibility. The framework systematically structures free-text eligibility criteria and computes similarity-based study weights reflecting population alignment between target and comparator trials. In a gastric cancer landscape analysis, EligMeta reduced 4,044 candidate trials to 39 clinically relevant studies through rule-based filtering, recovering 13 guideline-cited trials while maintaining strict criterion compliance. In an olaparib adverse events meta-analysis across four trials, eligibility-aware weighting shifted the pooled risk ratio from 2.18 (95\% CI: 1.71–2.79) under conventional Mantel-Haenszel estimation to 1.97 (95\% CI: 1.76–2.20), demonstrating quantifiable impact of incorporating eligibility alignment. EligMeta bridges automated trial discovery with eligibility-aware meta-analysis, providing a scalable and reproducible framework for evidence synthesis in precision medicine.
\end{abstract}

\section{Introduction}
In evidence-based medicine, synthesizing findings from clinical trials is essential for informing trial design, clinical practice guidelines, and regulatory decision-making. By pooling results across studies, evidence synthesis supports assessment of treatment efficacy, safety, and patient outcomes. This need is especially pressing in oncology and other precision medicine domains, where patient populations are increasingly stratified by biomarkers and where novel therapies, including immunotherapies and targeted treatments, are rapidly evolving. In these settings, the central challenge is no longer simply aggregating evidence but ensuring that synthesized estimates reflect the characteristics of the specific populations to which clinical and regulatory decisions will be applied.

In current practice, evidence synthesis relies on systematic reviews and meta-analyses. Systematic reviews aim to identify, screen, and extract relevant studies from large and growing clinical trial corpora~\cite{dickersin1994systematic}. Registries such as ClinicalTrials.gov contain vast volumes of heterogeneous and partially structured data, where key elements including eligibility criteria, endpoints, and outcomes are often embedded in free text with inconsistent phrasing. While efficiently identifying and filtering studies according to complex clinical criteria such as biomarker-defined populations, trial phase, or endpoint definitions is increasingly feasible with modern natural language processing (NLP) approaches, these capabilities have not yet been integrated into meta-analytic frameworks that account for population differences. 

Meta-analysis, the core statistical component of evidence synthesis, also faces limitations. Conventional approaches, including fixed- and random-effects models~\cite{hedges1998fixed, jackson2016extending, inthout2014hartung}, weight studies primarily by statistical precision and do not explicitly account for clinical compatibility across trials, particularly differences in eligibility criteria that define study populations. Consequently, pooled estimates may reflect averages across heterogeneous populations and may be misaligned with a specific target cohort relevant for clinical decision-making or trial design. These challenges motivate the development of methods that jointly address systematic evidence retrieval and statistical integration across clinically heterogeneous studies, while enabling eligibility-aware meta-analysis that explicitly accounts for differences in study populations.

Recent progress in large language models (LLMs) has automated pieces of the systematic review pipeline. However, these efforts remain fragmented across different components of the evidence synthesis process and do not support eligibility-aware meta-analysis. Methods for clinical trial understanding, including rule-based and supervised systems such as CriteriaMapper~\cite{lee2024establishing}, transformer-based models such as BioBERT~\cite{lee2020biobert} and BioGPT~\cite{luo2022biogpt}, and more recent LLM-based frameworks such as AlpaPICO~\cite{ghosh2024alpapico}, AutoCriteria~\cite{datta2024autocriteria}, zero-shot trial patient matching~\cite{wornow2025zero}, Patient2Trial~\cite{datta2025patient2trial}, and TrialGPT~\cite{jin2024matching}, support tasks such as eligibility criteria extraction, trial patient matching, and structured interpretation of clinical trial text at scale. These approaches substantially reduce the effort required to process heterogeneous trial records, but are primarily designed for retrieval and extraction rather than end-to-end evidence synthesis. In parallel, automated meta-analysis pipelines such as SEETrials~\cite{lee2024seetrials} and related approaches~\cite{yun2024automatically, ahad2024empowering}, support numerical result extraction and statistical pooling. However, these methods rely on conventional precision-based weighting, typically inverse-variance, and do not incorporate clinically interpretable measures of population similarity. Complementary work on eligibility criteria representation including concept-based approaches leveraging unified medical language system~\cite{bodenreider2004unified} and clustering-based or embedding-based methods~\cite{hao2014clustering, bornet2025analysis}, enables trial comparison and grouping for applications such as protocol design and recruitment optimization. However, these representations have not been integrated into formal meta-analytic frameworks.

Addressing this gap requires an integrated approach that dynamically coordinates retrieval, reasoning, and statistical computation within a single framework. Emerging agentic systems that leverage LLMs to orchestrate multi-step workflows offer a natural paradigm for such integration~\cite{wang2025survey}. We introduce EligMeta, an eligibility-aware evidence synthesis framework that employs an agentic architecture to bridge automated trial discovery with meta-analysis that explicitly accounts for eligibility-based population heterogeneity. At its core is a novel meta-analytic model that incorporates eligibility criteria directly into study weighting, enabling cohort-specific pooled estimates aligned with a user-specified target population.

EligMeta operates end-to-end from free-text queries to quantitative synthesis. Starting from a user-provided clinical question, the framework retrieves relevant trial records from public registries, applies customizable inclusion and exclusion criteria, extracts structured trial-level information including eligibility criteria, interventions, and outcomes, and performs eligibility-aware meta-analysis in which each trial's contribution is modulated by an interpretable measure of population similarity to the target. To ensure reproducibility and transparency, EligMeta employs a hybrid architecture: LLMs provide high-level reasoning and orchestration for natural language understanding and workflow planning, while numerically critical operations including trial selection, eligibility-weight computation, and statistical estimation are executed through deterministic, version-controlled modules. This design ensures complete provenance and auditability from unstructured inputs to quantitative outputs, yielding pooled estimates that are both statistically robust and clinically meaningful.

Our contributions are threefold. First, we introduce a novel eligibility-aware meta-analysis framework that quantifies eligibility criteria similarity between candidate trials and a user-specified target trial, incorporating these measures into statistical pooling to produce population-specific estimates. To our knowledge, this is the first meta-analytic framework to explicitly incorporate eligibility criteria into study weighting. Second, we develop an end-to-end agentic pipeline that couples automated trial discovery from large registries with criteria-based filtering, structured extraction, and eligibility-aware statistical synthesis, bridging the gap between advances in clinical NLP and formal meta-analytic methods. Third, we propose a hybrid architecture that strategically separates LLM-based reasoning from deterministic computation, addressing a key challenge in agentic systems: balancing the flexibility of LLMs with the reproducibility requirements of statistical analysis.

\section{Methods}
EligMeta transforms free-text clinical queries into cohort-specific, reproducible meta-analytic estimates through a structured, multi-stage pipeline. Current agentic systems face challenges in balancing flexibility with reproducibility. Pre-defined function-calling approaches~\cite{schick2023toolformer} ensure consistency but cannot flexibly handle the heterogeneity and complexity of clinical trial registries, while end-to-end agentic coding systems~\cite{openai2025codex, anthropic2026claudecode} offer adaptability but at the cost of reproducibility and computational efficiency. EligMeta addresses these limitations through a hybrid architecture that combines LLM-based reasoning with deterministic execution of numerically critical operations. As illustrated in Figure~\ref{fig:figure1}, the framework operates in two stages. In the first stage, clinical queries are translated into structured rule sets defining inclusion criteria and a target eligibility profile, which are then applied through deterministic evaluation to retrieve and filter relevant trials from public registries. In the second stage, eligibility criteria are used to quantify population similarity relative to the target, and these similarity measures are then incorporated into statistical pooling to produce eligibility-weighted estimates. The architecture and implementation of each component are described in the subsections that follow. 

\begin{figure}[htbp!]
    \centering
    \includegraphics[width=0.95\linewidth]{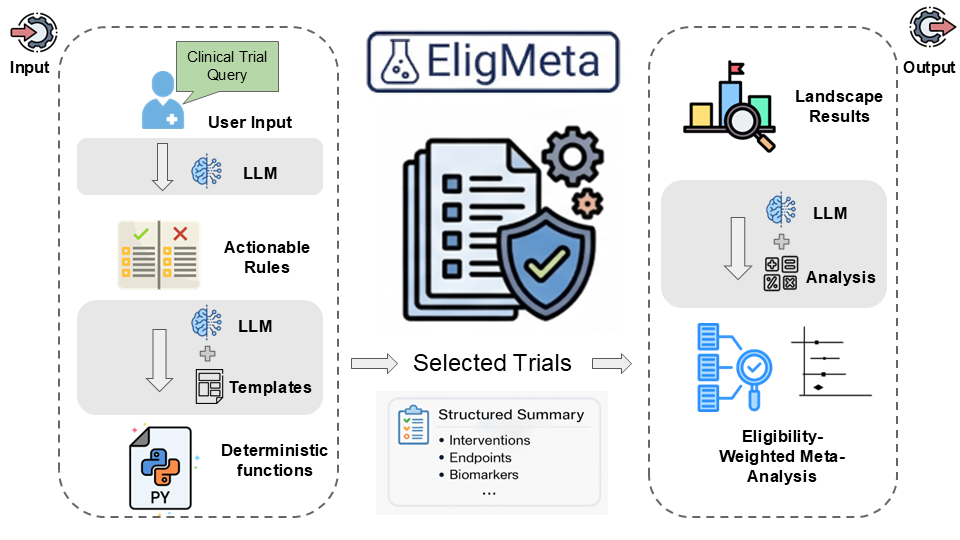}
    \caption{Overview of the EligMeta framework for clinical trial evidence synthesis. A free-text clinical query is translated into structured, auditable rule sets defining study selection criteria and a target eligibility profile. These rules are applied through deterministic evaluation to retrieve and filter relevant trials from public registries, with all intermediate artifacts available for expert review. The resulting studies support downstream landscape analysis and eligibility-aware meta-analysis, in which each study's contribution to the pooled estimate is weighted by its population similarity to the target trial.}
    \label{fig:figure1}
\end{figure}

\subsection{Trial Selection and Structuring}
The first stage of EligMeta transforms a free-text clinical query into an explicit, auditable process for identifying relevant trials from clinical trial registries (Figure~\ref{fig:figure2}). We implement this stage using ClinicalTrials.gov as the primary data source, though the framework can be generalized to other registry and literature databases. As an illustrative example, consider the query: “Identify and evaluate clinical trials of gastric cancer or gastroesophageal junction cancer. Trials must investigate targeted therapies or immunotherapies; report survival outcomes such as progression-free survival (PFS) or overall survival (OS); and enroll biomarker-stratified populations (e.g., HER2-positive, MSI-H, PD-L1 positive). Exclude Phase 3 trials with a small sample size (e.g., fewer than 100 enrolled patients). Include only trials evaluating FDA-approved drugs.” This type of query is representative of clinical evidence synthesis tasks, where investigators define a target population, specify interventions and outcomes of interest, and apply multiple inclusion and exclusion criteria. 

\begin{figure}[htbp!]
    \centering
    \includegraphics[width=0.95\linewidth]{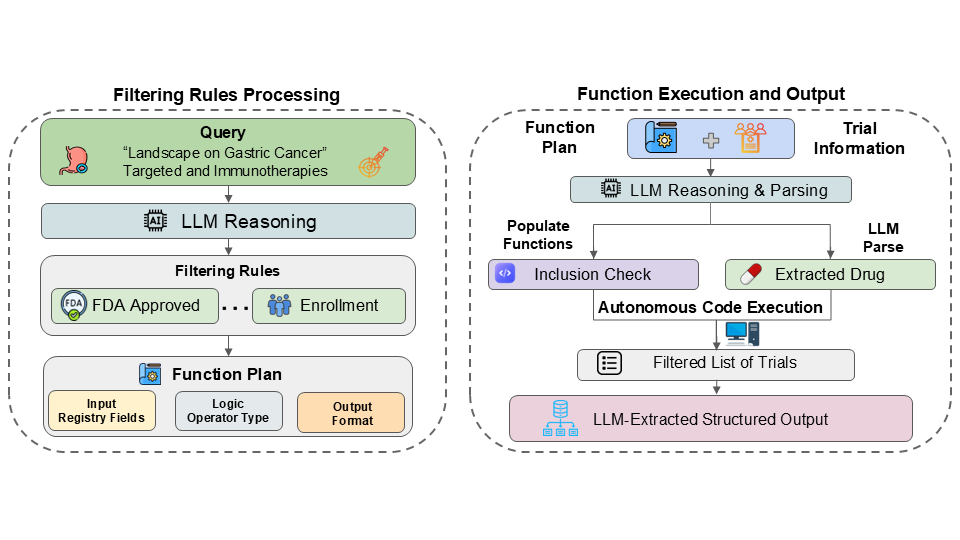}
    \caption{Trial selection and structuring workflow. Left: A free-text clinical query is translated into a reviewed set of structured rules defining study selection criteria and target conditions. Each rule is transformed to a function plan via LLM reasoning. Right: The function plans are applied to trial records using format-constrained parsing and deterministic evaluation, producing a filtered set of trials and structured summaries. LLM usage is restricted to rule specification and parsing, while all selection operations are deterministic and auditable.}
    \label{fig:figure2}
\end{figure}

EligMeta processes this input through two sequential steps that produce intermediate outputs supporting inspection and expert revision. In the first step, the query is translated into a concise, human-readable rule set. Each rule encodes a single selection criterion, such as disease type (e.g., gastric or gastroesophageal cancer), intervention class (e.g., targeted therapy or immunotherapy), endpoint requirement (e.g., PFS or OS), biomarker stratification, trial phase, enrollment threshold, or regulatory status. Collectively, these rules define both the study selection criteria and the characteristics of the target population. The rule set is surfaced for expert review prior to execution, allowing conditions to be inspected and refined to ensure alignment with clinical intent. Additional details on rule generation are provided in Supplementary Section A.1.

In the second step, each rule is operationalized through a paired evaluation template, forming a function plan that specifies how the rule is applied to trial data. Each function plan defines the registry fields to be examined, the parsing instruction for extracting relevant information, the comparison operator, and the expected output type. For unstructured fields such as eligibility criteria and study descriptions, LLMs extract structured values using minimal, format-constrained prompts. All extracted outputs conform to predefined schemas (e.g., boolean indicators, numeric values, or controlled categorical labels) and are subsequently evaluated through deterministic functions. Further details are provided in Supplementary Section A.2. FDA approval status, a common filtering criterion, is inherently dynamic and requires validation against reliable external sources. We address this through an agentic retrieval module (Supplementary Section A.3) that queries Drugs.com as a curated, versioned reference to determine approval status for each drug under investigation.  

Trial selection begins with generic pre-filtering conditions common to clinical trial meta-analysis: inclusion of interventional studies with completed status, exclusion of Phase IV or missing-phase trials, and requirement of available results or publications. EligMeta then applies the rule-specific functions sequentially to the retrieved trial records from ClinicalTrials.gov, yielding a filtered set of studies that satisfy all specified criteria. The resulting studies are compiled into structured summaries capturing key trial attributes, including interventions, biomarkers, endpoints, and study design characteristics. These standardized representations support downstream landscape analysis and serve as input to the eligibility-aware meta-analysis stage. All intermediate outputs, including rule sets, function plans, parsed field values, and trial-level filtering outcomes, are logged to provide a complete audit trail from the original query to the final selected studies. Additional implementation details are provided in Supplementary Section A.4.

\subsection{Eligibility-Aware Meta-Analysis}
The trial selection stage yields a curated set of studies with structured trial-level information. When the objective is quantitative synthesis for a prespecified target population, EligMeta performs eligibility-aware meta-analysis that explicitly accounts for differences in study populations. The central idea is to quantify the alignment between each candidate trial and the target trial using eligibility criteria, and to modulate each study's contribution to the pooled estimate accordingly. Trials whose eligibility profiles more closely match the target trial receive greater weight, while those with substantial differences are down-weighted.

As illustrated in Figure~\ref{fig:figure3}, this stage takes the selected trials as input, with their free-text eligibility criteria, and the eligibility criteria of a user-specified target trial that defines the reference population. An LLM derives explicit, human-readable penalty rules that encode clinically meaningful deviations between each candidate trial's population and the target. To enable deterministic evaluation, the free-text eligibility criteria of all trials are converted into a structured representation in which each criterion is standardized into components such as entity, attribute, and value; details of this structured representation are provided in Supplementary Section B.1.

\begin{figure}[htbp!]
    \centering
    \includegraphics[width=0.95\linewidth]{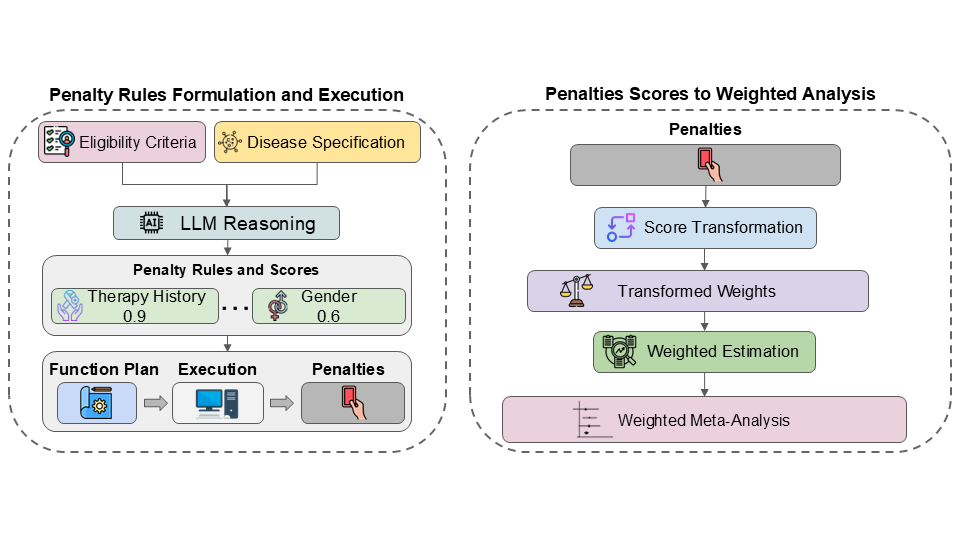}
    \caption{Eligibility-aware meta-analysis workflow. Left: Free-text eligibility criteria, together with a target disease specification, are translated into explicit, human-readable penalty rules that capture clinically meaningful differences between trial populations and the target scenario. These rules are evaluated on a structured representation of eligibility criteria to produce trial-level penalty scores through deterministic operations. Right: Penalty scores are transformed into eligibility weights, which are then incorporated into weighted meta-analysis to produce population-aligned effect estimates. LLM usage is restricted to penalty rule specification and schema-constrained parsing, while penalty evaluation, score transformation, and statistical estimation are deterministic and auditable.}
    \label{fig:figure3}
\end{figure}

Penalty rules are evaluated deterministically on this structured representation. Each rule operates at the criterion level, targeting specific eligibility clauses rather than entire free-text blocks, ensuring that every triggered penalty is directly attributable to an identifiable population difference. A penalty is triggered when at least one relevant criterion satisfies the rule condition. For each trial, the total penalty score is computed by aggregating all triggered penalties, reflecting the cumulative degree of population mismatch relative to the target. A predefined transformation then maps the penalty score $p_i$ to a non-negative eligibility weight$w_i$, subject to $w_i>0$ and $\sum_i w_i =1$, where the transformation is governed by tunable parameters that control the baseline influence and the steepness of down-weighting (Supplementary Section B.2). Because both the penalty rules and their scores are derived relative to the selected trials and the target trial, the resulting weights are case-specific rather than fixed. This criterion-level design improves interpretability and enables precise audit of population discrepancies across studies.

To illustrate this approach, we present an eligibility-weighted Mantel–Haenszel (EW-MH) estimator for binary outcomes~\cite{mantel1959statistical}. The classical Mantel–Haenszel (MH) estimator combines study-specific effect estimates under a common-effect assumption using precision-based weights, without accounting for population alignment with a target trial. We extend this framework by incorporating eligibility weights $w_i$ to adjust each study's contribution according to its population similarity to the target. The EW-MH estimator is defined as:

\begin{equation}
\hat{\theta}_{EW-MH} = \frac{\sum_{i=1}^{k} w_i \frac{a_i n_{0i}}{n_i}}{\sum_{i=1}^{k} w_i \frac{a_i n_{1i}}{n_i}}
\end{equation}

where $a_i$ is the number of events in the treatment arm, $c_i$ is the number of events in the control arm, $n_{0i}$ is the total control sample size, $n_{1i}$ is the total treatment sample size, and $n_i=n_{0i}+n_{1i}$ is the total sample size. When $w_i=1/k$ , Equation (1) reduces to the classical MH estimator. The eligibility-weighted formulation assigns greater influence to studies whose eligibility criteria align more closely with the target trial, yielding pooled estimates more relevant to the intended clinical context. Unbiasedness proof and variance derivation are provided in Supplementary Section B.3. While we focus on the EW-MH estimator for concreteness, the eligibility weighting framework can be incorporated into alternative fixed- or random-effects models, or Bayesian formulations. All steps in this stage, including penalty evaluation, score transformation, and weighted estimation, are executed deterministically and fully logged.

\section{Results}
\subsection{Illustrative Example: Trial Selection for Gastric Cancer}
To illustrate the trial selection and structuring stage, we apply EligMeta to the gastric cancer query introduced in Section 2.1. EligMeta automatically generates six inclusion and exclusion rules capturing disease type, intervention class, survival endpoints, biomarker stratification, enrollment constraints, and regulatory status. The initial query to ClinicalTrials.gov returns 4,044 trials. After generic pre-filtering, the candidate set is reduced to 700 trials. The six rules are then applied sequentially through deterministic evaluation, yielding a final set of 39 eligible studies. The PRISMA-style flow diagram (Figure~\ref{fig:figure4}) provides a transparent account of trial inclusion and exclusion at each step, with each filtering decision directly attributable to specific rule evaluations. Additional details are provided in Supplementary Section C.

\begin{figure}[htbp!]
    \centering
    \includegraphics[width=0.95\linewidth]{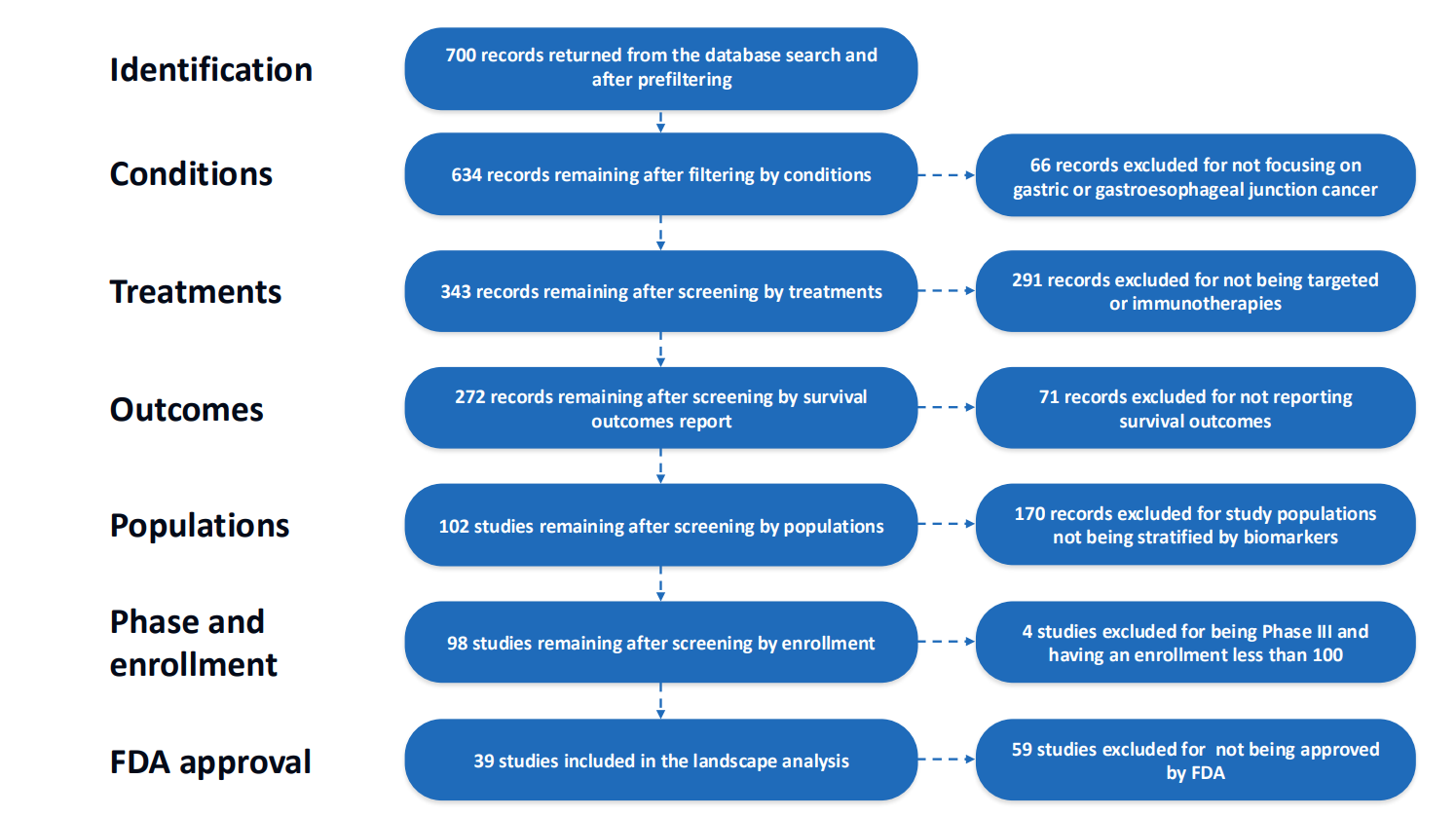}
    \caption{Stepwise filtering results for the gastric cancer use case. The diagram follows the six rules described in the text, showing the number of clinical trials remaining and excluded at each filtering step.}
    \label{fig:figure4}
\end{figure}

Table~\ref{tab: representive_five} presents five representative trials from selected trials in a structured evidence format containing standardized fields, including interventions, biomarkers, study phase, enrollment size, and survival endpoints, that support downstream evidence review and quantitative synthesis. Complete outputs for all 39 selected trials are available in the project repository. The selected trial set demonstrates strong clinical relevance and alignment with expert-curated evidence. Among the 39 trials, 13 are explicitly cited in the National Comprehensive Cancer Network (NCCN) guideline~\cite{nccn_gastric_2026} targeted therapy or immunotherapy sections, indicating that EligMeta recovers guideline-level evidence directly from registry data using only a free-text query, without requiring prior knowledge of study names, NCT numbers, or publication references. The remaining 26 trials comprise early-phase studies (Phase I/II) or trials with incomplete reported data that may not yet be incorporated into clinical guidelines but represent emerging evidence in the field.

\begin{table}[htbp!]
\centering
\tiny
\caption{Representive eligible trials returned by the landscape stage of EligMeta.}
\label{tab: representive_five}
\begin{adjustbox}{max width=\textwidth}
\begin{tabularx}{\textwidth}{X X X X X X X X}
\toprule
\makecell[l]{\textbf{NCT} \\ \textbf{Number}} &
\textbf{Intervention(s)} &
\textbf{Biomarker} &
\textbf{Condition} &
\makecell[l]{\textbf{Study} \\ \textbf{Phase}} &
\makecell[l]{\textbf{Enrollment} \\ \textbf{Size}} &
\textbf{Status} &
\makecell[l]{\textbf{Endpoints:} \\ \textbf{PFS} \\ }\\
\midrule

NCT03615326 &
Pembrolizumab vs Trastuzumab &
HER2 &
Gastric Cancer &
PHASE3 & 738 & COMPLETED &
median PFS 10.0 mo vs 8.1 mo, HR 0.73 (95\% CI 0.61-0.87), p=0.0002 \\

\midrule
NCT03329690 &
DS-8201a vs Irinotecan vs Paclita &
HER2 &
gastric cancer, gastroesophageal junction adenocarcinoma &
PHASE2 & 233 & COMPLETED &
median PFS 5.6 mo (95\% CI 4.3-6.9) \\

\midrule
NCT03504397 &
zolbetuximab vs placebo &
Claudin 18.2, HER2-Negative &
stomach cancer, GEJ cancer &
PHASE3 & 565 & ACTIVE NOT  RECRUITING &
median PFS 11.04 mo vs 8.94 mo, HR 0.734 (95\% CI 0.591-0.910), p=0.0024 \\

\midrule
NCT03653507 &
zolbetuximab vs placebo &
Claudin 18.2, HER2 &
gastric cancer, GEJ cancer &
PHASE3 & 507 & ACTIVE NOT RECRUITING &
median PFS 8.21 mo vs 6.80 mo, HR 0.687 (95\% CI 0.544-0.866), p=0.0007 \\

\midrule
NCT04014075 &
Trastuzumab deruxtecan &
HER2 &
gastric cancer, gastroesophageal junction cancer &
PHASE2 & 379 & COMPLETED &
median PFS 5.5 mo (95\% CI 4.2-7.3) \\
\bottomrule
\end{tabularx}
\end{adjustbox}
\vspace{2pt}
\raggedright
\footnotesize \textit{Abbreviations:} PFS, progression-free survival; GEJ, gastroesophageal junction; CI, confidence interval; HR, hazard ratio.
\end{table}

The selected trials span multiple biomarker classes (HER2-positive, MSI-H, PD-L1 CPS thresholds), intervention types (monoclonal antibodies, checkpoint inhibitors, targeted agents), and disease settings (locally advanced, metastatic, adjuvant), illustrating the framework's capacity to navigate complex clinical stratification schemes embedded in heterogeneous free-text eligibility criteria. 

To assess the performance of EligMeta relative to general-purpose AI systems, we conducted two parallel gastric cancer landscape analyses using (i) GPT-5.4 with Deep Research~\cite{openai2026gpt54}, representing a general-purpose LLM-based research assistant,  and (ii) Codex~\cite{openai2025codex}, representing an agentic coding paradigm for automated pipeline construction. Both analyses were conducted using the same query; complete implementation details and comparison outputs are provided in Supplementary Section D.  GPT-5.4 only identified 11 trials, of which 8 were covered by the NCCN guidelines, whereas Codex identified 28 trials, of which only 7 were covered. Although both methods captured some pivotal trials, such as KEYNOTE-811 (NCT03615326) and DESTINY-Gastric01 (NCT03329690), they missed several important studies, including KEYNOTE-059 (NCT02335411), KEYNOTE-859 (NCT03675737), and RATIONALE-305 (NCT03777657). In addition, the Codex output included the FLX475 trial (NCT04768686), despite FLX475 not being FDA-approved, thereby violating the predefined inclusion criterion restricting trials to those evaluating FDA-approved therapies.

\subsection{Eligibility-Aware Meta-Analysis: Olaparib Adverse Events}
To demonstrate the eligibility-aware meta-analysis framework, we apply EligMeta to the study of Ricci et al.~\cite{ricci2020specific}, which conducted a systematic review and meta-analysis of maintenance olaparib versus placebo in advanced malignancies. This example provides a controlled setting to evaluate how eligibility-aware weighting affects pooled estimates compared to conventional precision-based meta-analysis. Using the query-driven pipeline described in Section 3.1, EligMeta identifies five eligible trials, including all four randomized controlled trials analyzed in the original study: Golan et al.~\cite{golan2019maintenance}, Moore et al.~\cite{moore2018maintenance}, Ledermann et al.~\cite{ledermann2014olaparib}, and Pujade-Lauraine et al.~\cite{pujade2017olaparib}. We focus on these four studies to enable direct comparison with the published meta-analysis. The eligibility criteria of all four trials are provided in Supplementary Section E.

\paragraph{Eligibility-based weighting.} We designate Golan et al.~\cite{golan2019maintenance} as the target trial and compute penalty scores quantifying each study’s eligibility mismatch relative to this target. The LLM generates five interpretable penalty rules, each capturing a clinical meaningful population difference:
\begin{enumerate}[label=\textbf{R\arabic*.}]
  \item Score (0.9): 
        Assign a high penalty for prior trials that require patients to have completed multiple lines of platinum-based therapy, whereas the target trial includes patients receiving initial chemotherapy for metastatic disease.
  \item Score (0.7): 
        Assign a moderate penalty for prior trials that only include female patients with ovarian, fallopian tube, or primary peritoneal cancer, while the target trial focuses on pancreas adenocarcinoma.
  \item Score (0.6): 
        Assign a moderate penalty for prior trials that require patients to have a stable or maintained response to their last chemotherapy course, whereas the target trial allows patients without evidence of disease progression.
  \item Score (0.6): 
        Assign a moderate penalty for prior trials that require randomization within a specific timeframe after the last chemotherapy dose, which is not specified in the target trial.
  \item Score (0.5): 
        Assign a moderate penalty for prior trials that exclude patients with stable disease on post-treatment scans, while the target trial includes patients with measurable or non-measurable disease.
\end{enumerate}

Evaluating these rules deterministically on structured eligibility criteria yields total penalty scores:\[
\bigl(p_{\text{Golan}},\;p_{\text{Moore}},\;p_{\text{Ledermann}},\;p_{\text{Pujade}}\bigr)
  \;=\;(0.0,\,2.8,\,1.8,\,2.8),
\]. The target trial incurs zero penalty by definition, while other studies are penalized according to their eligibility mismatch. To incorporate these penalties into the meta-analysis, we transform the penalty scores into eligibility weights using the mapping described in Section 2.2 and detailed in Supplementary Section B. With parameters $\gamma=0.5$ and $B=20$, this yields normalized eligibility weights: 
\begin{equation}
\bigl(w_{\text{Golan}},\;w_{\text{Moore}},\;w_{\text{Ledermann}},\;w_{\text{Pujade}}\bigr)
  \;=\;(0.52,\,0.13,\,0.21,\,0.13).
  \label{eq: computed_weights}
\end{equation}
Where $\sum_i w_i = 1$. These weights modulate each study's contribution to the pooled estimate according to its population alignment with the target.

The eligibility weights are then combined with conventional precision-based weights in the EW-MH estimator (Equation 1, Section 2.2), which incorporates both eligibility weights (from population alignment) and precision-based weights (from sample size and event rates). To visualize and compare each study's contribution under both approaches, we compute display weights, which represent the percentage contribution of each study normalized to sum to 100\%. For the conventional MH estimator, the display weight for trial i is:
\begin{equation*}
    w_{i,MH}^{disp}=\frac{\frac{c_i n_{1i}}{n_i}}{\sum_j \frac{c_j n_{1j}}{n_j}}\times100,
\end{equation*}
where $c_i$ is the number of control events, $n_{1i}$ is the treatment sample size, and $n_i$ is the total sample size. For the eligibility-weighted MH estimator, the display weight incorporates the eligibility weight $w_i$:
\begin{equation}
    w_{i,EW}^{disp}=\frac{w_i\frac{c_i n_{1i}}{n_i}}{\sum_j w_j \frac{c_j n_{1j}}{n_j}}\times100.
\end{equation}
When all eligibility weights equal 1, the two formulas are identical. Table~\ref{tab: ew_weights} presents the display weights under both approaches. Under eligibility-aware weighting, Golan 2019's display weight increases from 13.6\% to 34.6\% due to its perfect alignment with the target population, while Moore 2018 decreases from 31.7\% to 20.5\% due to substantial eligibility mismatch.

\begin{table}[htbp!]
  \centering
  \caption{Olaparib maintenance treatment vs placebo, outcome: risk ratio (RR) of vomiting (all-grade). $w_{i,MH}^{disp}$: conventional precision-based display weights from Ricci et al.~\cite{ricci2020specific}. $w_{i,EW}^{disp}$: eligibility-weighted display weights using Golan et al.~\cite{golan2019maintenance} as target population. CI: 95\% confidence interval.}
  \label{tab: ew_weights}
  \scriptsize
  \begin{tabular}{@{}lrrrrrrrrr@{}}
     &\multicolumn{2}{c}{\textbf{Olaparib}} & \multicolumn{2}{c}{\textbf{Placebo}} & \\
    \toprule
    \textbf{Study} & \textbf{Events} & \textbf{Total} & \textbf{Events} & \textbf{Total} &
     \textbf{RR} &
    \textbf{CI$_{\text{low}}$} & \textbf{CI$_{\text{high}}$} &
    \textbf{$w_{i,MH}^{disp}$} &
    \textbf{$w_{i,EW}^{disp}$} \\
    \midrule
    Golan 2019            &  18 &  91 &  9 &  60  & 1.32 & 0.63 & 2.74 & 13.6 &34.6 \\
    Moore 2018            & 104 & 260 & 19 & 130 & 2.74 & 1.76 & 4.26 & 31.7 & 20.5 \\
    Ledermann 2014        &  43 & 136 & 18 & 128 & 2.25 & 1.37 & 3.69 & 23.2 & 24.4 \\
    Pujade–Lauraine 2017  &  73 & 195 & 19 &  99 & 1.95 & 1.25 & 3.04 & 31.5 & 20.4 \\
    \midrule
    \multicolumn{5}{@{}l}{\textit{Total (original result from \citet{ricci2020specific})}} &
    2.18 & 1.71 & 2.79 & & \\
    \multicolumn{5}{@{}l}{\textit{Total (eligibility‐weighted, ours)}} &
    1.97 & 1.76 & 2.20 & & \\
    \bottomrule
  \end{tabular}
\end{table}

\paragraph{Impact on pooled estimates.} We analyze the outcome of all-grade vomiting, with event data and pooled estimates presented in Table~\ref{tab: ew_weights} and Figure~\ref{fig:figure5}. Under classical precision-based MH weighting, the pooled risk ratio is $\hat{\theta}_{MH}$=2.18 (95\% CI 1.71–2.79). Under eligibility-aware weighting: $\hat{\theta}_{EW-MH}$=1.97 (95\% CI 1.76–2.20). This shift reflects a fundamental difference in how evidence is synthesized. The EW-MH estimator assigns greater influence to studies whose eligibility criteria align more closely with the target population, while down-weighting studies with greater eligibility mismatch. In this example, the lowest-risk trial (Golan 2019, RR = 1.32) becomes the dominant contributor due to its alignment with the target scenario. Conversely, Moore 2018, which has both a higher risk ratio and greater eligibility mismatch, contributes less to the pooled estimate.

\begin{figure}[htbp!]
    \centering
    \includegraphics[width=0.95\linewidth]{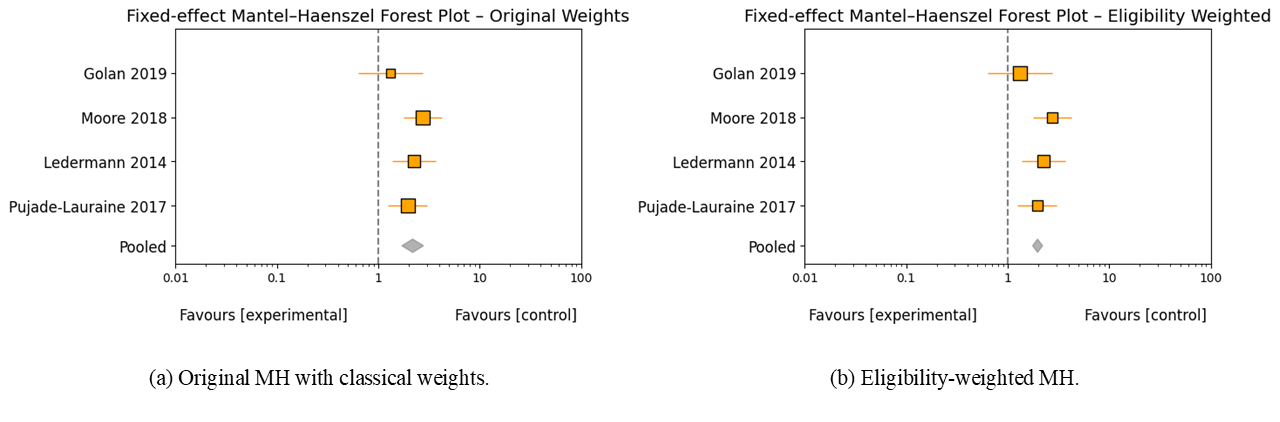}
    \caption{Comparison of pooled risk ratio estimates for all-grade vomiting under classical precision-based weighting (left) and eligibility-aware weighting (right), for olaparib maintenance versus placebo.}
    \label{fig:figure5}
\end{figure}

This example demonstrates that conventional precision-based meta-analysis may produce pooled estimates dominated by statistically precise but clinically misaligned studies. Moore 2018, while contributing 31.7\% under classical weighting due to large sample size, represents a fundamentally different patient population (ovarian cancer, multiple prior lines of therapy) than the target pancreatic cancer scenario in Golan 2019. Eligibility-aware weighting explicitly accounts for this mismatch, producing a pooled estimate more representative of the target population. Critically, this adjustment is fully interpretable: each study's down-weighting can be traced to specific eligibility criteria differences (prior treatment lines, cancer type, response requirements), providing transparent rationale for the reweighted synthesis. While the eligibility weighting framework is demonstrated here using the EW-MH estimator, it can be incorporated into alternative fixed- or random-effects models or Bayesian formulations. By integrating structured eligibility information into meta-analysis, EligMeta extends classical methods to produce evidence summaries that are both statistically principled and explicitly aligned with target clinical populations.

\section{Discussion}
We present EligMeta, an eligibility-aware framework for clinical trial evidence synthesis that integrates structured rule-based trial selection with population-aligned meta-analysis. By combining LLM-assisted reasoning with deterministic execution, the framework translates free-text clinical queries into explicit, auditable workflows and produces reproducible, clinically meaningful statistical estimates.

A key contribution of this work is the introduction of eligibility-aware weighting in meta-analysis. Conventional meta-analysis weights studies primarily by statistical precision, implicitly assuming that all included trials target comparable populations. EligMeta challenges this assumption by explicitly quantifying population similarity through structured eligibility criteria and adjusting each study's contribution accordingly. As demonstrated in Section 3.2, eligibility-aware weighting can produce meaningfully different pooled estimates: the risk ratio for olaparib-related vomiting shifted from 2.18 to 1.97 when studies were weighted by alignment with a pancreatic cancer target population rather than by precision alone. Critically, this adjustment is fully interpretable: each study's influence can be traced to specific eligibility differences such as prior treatment lines, cancer type, and response requirements, providing transparency rarely achieved in traditional meta-analytic pipelines.

Beyond methodological innovation, EligMeta establishes a reproducible and auditable workflow for evidence synthesis. The framework constrains LLM usage to rule formulation and schema-constrained parsing while executing all selection, weighting, and statistical computations deterministically. This design reduces hallucination risk while maintaining flexibility in handling complex clinical queries, enabling integration of language models into high-stakes biomedical applications without sacrificing reliability. LLM-driven components enable adaptive reasoning over heterogeneous eligibility text, while structured schemas and deterministic scoring ensure consistent, reproducible weighting. Compared to common accessible AI deep research tools or agentic coding clients, EligMeta provides a more reliable and task-aligned approach for meta-analysis with significantly lower reasoning cost through more efficient LLM models.

Several limitations and future directions warrant consideration. First, the framework depends on accurate extraction and structuring of eligibility criteria, which may be limited by inconsistencies in trial reporting quality. Extending the structured representation to capture more nuanced eligibility logic (e.g., nested conditions, time-dependent criteria) would improve fidelity. Second, the current implementation relies on curated registries such as ClinicalTrials.gov. Future work will integrate direct literature retrieval to recover information often missing from registries, including numbers at risk, follow-up windows, stratification details, complete adverse event tables with grade and timing, protocol amendments, subgroup effect estimates, and additional secondary outcomes. Linking peer-reviewed publications to registry entries would reduce missingness and selective reporting while enabling living updates as new evidence emerges. Third, the choice of penalty rules, severity scores, and transformation parameters introduces modeling decisions that may influence results. We recommend sensitivity analyses across parameter ranges and incorporation of domain expert input to calibrate penalty scores. The interpretable rule structure facilitates expert review and iterative refinement. Lastly, while we demonstrate the eligibility weighting framework using binary outcomes and the Mantel-Haenszel estimator, extension to continuous outcomes, time-to-event data, and alternative meta-analytic models (e.g., random-effects, Bayesian hierarchical models) would broaden applicability.

\bigskip

\noindent\textbf{Code Availability:} The code used in this study is available at: \url{https://github.com/JackZhao1998/EligMeta.git}

\vspace{4pt}
\noindent\textbf{Data Availability:} All datasets analyzed in this study are publicly available from \url{https://clinicaltrials.gov/}.

\bibliography{main_paper}

\newpage
\clearpage
\setcounter{page}{1}
\setcounter{section}{0}
\setcounter{subsection}{0}
\setcounter{figure}{0}
\setcounter{table}{0}

\setcounter{page}{1}
\renewcommand{\thesection}{\Alph{section}}
\renewcommand{\thesubsection}{\Alph{section}.\arabic{subsection}}

\renewcommand{\thefigure}{S\arabic{figure}}
\setcounter{figure}{0} 

\renewcommand{\thetable}{S\arabic{table}}
\setcounter{table}{0} 

\begin{center}
    {\LARGE Supplementary Materials for ``Eligibility-Aware Evidence Synthesis: An Agentic Framework for Clinical Trial Meta-Analysis"}
\end{center}

\section{EligMeta Implementation Details}
\subsection{LLM-based Rule Generation}
To enable systematic and reproducible landscape analysis, we developed a rule generation module that uses carefully designed prompts to guide a large language model (LLM), described in Section 2.1 of the main text. Rather than relying on ad hoc keyword matching or manual filtering, this component operationalizes clinical intent into explicit inclusion and exclusion criteria that can be consistently applied across trial databases.

The prompt is designed to mimic the reasoning process of a clinical trial analyst. It guides the language model to interpret an under-specified research question (e.g., comparing therapeutic strategies or exploring treatment patterns within a disease area) and reformulate it into a comprehensive yet concise set of actionable rules. These rules are expressed in natural language and follow a standardized format, ensuring interpretability for domain experts while remaining amenable to downstream computational processing. 

Several key design principles were incorporated:
\begin{itemize}
    \item \textbf{Clinical relevance prioritization.} The prompt explicitly instructs the model to ground rule generation in standard clinical trial design practices. This includes identifying appropriate trial phases, intervention types, endpoints, and comparator structures that are commonly used in the disease area of interest.

    \item  \textbf{Separation of study-level and patient-level criteria.} To maintain consistency with evidence synthesis workflows, the prompt restricts rule generation to study-level selection logic (e.g., trial design, interventions, outcomes), while excluding patient-level eligibility criteria (e.g., age, laboratory values), which are typically handled separately in meta-analysis protocols.

    \item \textbf{Generalizable disease and treatment framing.} The prompt encourages broad but precise phrasing (e.g., "trials that investigate the disease of interest" rather than overly restrictive subtypes), unless the user explicitly specifies narrower conditions. This reduces the risk of unintentionally excluding relevant studies.

    \item \textbf{Avoidance of redundant or pre-filtered criteria.} The prompt explicitly excludes rules related to study completion status or publication availability, as these are handled upstream through structured queries to ClinicalTrials.gov. This ensures that generated rules focus only on trial selection dimensions that require complex semantic reasoning.
\end{itemize}

The complete prompt used for rule generation is presented below.

\begin{tcolorbox}[
    colback=blue!5,
    colframe=blue!75!black,
    title={Rule Generation Prompt},
    fonttitle=\bfseries,
    boxrule=0.6pt,
    arc=2pt,
    left=6pt,
    right=6pt,
    top=6pt,
    bottom=6pt
]
\small

You are a biomedical trial analyst helping define selection logic for landscape analysis of clinical trials.

You are given:
\begin{itemize}
    \item A vague user query describing a clinical research interest
    \item A disease context (e.g., lung cancer, breast cancer)
\end{itemize}

Your job is to infer a \textbf{comprehensive set of rules} for selecting trials relevant to this question. Usually, the rules are clear and simple.

Each rule should:
\begin{itemize}
    \item Be specific and clearly actionable (e.g., ``Include only phase III trials'')
    \item Be based on standard clinical trial practice in the given disease
    \item Include trial design features (e.g., control arms, endpoints, blinding) where relevant
    \item Only generate rules that determine whether a clinical trial should be included in the meta-analysis --- do not describe internal trial eligibility
\end{itemize}

Do not include any of the following rules, as those are already taken care of in the pre-filtering step:
\begin{itemize}
    \item Check if the study has been completed
    \item Whether the study includes publications
\end{itemize}

\textbf{Important guidance:}
\begin{itemize}
    \item If the user query compares multiple treatments (e.g., pembrolizumab monotherapy vs.\ combination therapy vs.\ chemotherapy), then any two of them in the comparison is fine; all arms do not need to appear in the same trial.
    \item Usually, unless contradicted by the user intent or expert comment, exclude trials that do not involve the treatment of interest and trials that do not involve the disease of interest, to avoid false positives.
    \item When defining disease or treatment rules, prefer general phrasing unless the user explicitly specifies narrower criteria.
    \item Prioritize wording such as ``study'', ``deal with'', or ``involves'', rather than overly restrictive phrasing such as ``focus on''.
\end{itemize}

\textbf{Output requirements:}
\begin{itemize}
    \item Output a JSON list of plain-English rules
    \item Use phrases like ``Include trials that\ldots'' or ``Exclude trials that\ldots''
    \item Do not use markdown, explanation, or backticks
\end{itemize}
\end{tcolorbox}

\subsection{Function Plan Generation}
To operationalize the generated selection rules into executable logic, we designed a module that converts each plain-English rule into a structured filtering plan. While the previous step focuses on what trials should be included or excluded, this stage defines how that decision is implemented programmatically. The core idea is to decompose each rule into a set of minimal, machine-executable components that can be applied to trial metadata in a consistent and reproducible manner.

Instead of directly generating code, the system produces an intermediate representation as a structured plan that separates information extraction from logical evaluation. Introducing this planning step is essential to ensure reliability, interpretability, and consistency in the execution process. Directly generating executable code from natural-language rules often leads to brittle implementations, where parsing logic, comparison operations, and control flow are entangled and difficult to verify or reuse. By separating rule interpretation into a structured intermediate plan, the system enforces a clear decomposition between information extraction and logical evaluation, enabling modular execution and easier debugging.

The prompt design emphasizes three core principles: (1) modularity: each rule is decomposed into minimal, well-defined conditions to ensure that information extraction and logical evaluation remain explicitly separated; (2) efficiency: it prioritizes selective use of relevant metadata fields and supports structured logical operators (e.g., sequential evaluation) to enhance computational efficiency; and (3) execution fidelity: the design ensures that generated plans are interpretable, consistent, and directly translatable into reliable downstream execution.

The output JSON defines a structured execution plan for downstream filtering. Specifically, \texttt{filter\_name} provides a unique identifier for the generated function, while \texttt{logical\_operator} determines how multiple condition checks should be combined during execution. The \texttt{conditions} field contains the core operational units, where each condition specifies which trial metadata fields should be attended to, what the LLM parser should extract, and how the extracted result should be compared against a predefined target. This design allows the downstream system to first query only the relevant metadata, then apply standardized comparison logic in a deterministic manner. For conditions involving predefined membership sets, \texttt{membership\_list\_name} further enables execution outside the LLM, ensuring that list-based matching remains explicit and reproducible. 

The plan generation is applied to each individual rule from the previous list. The resulting plans are then aggregated with two additional pieces of information—(i) the condition (i.e., the disease or indication under study) and (ii) the treatment (i.e., the intervention or drug being investigated), both parsed by the LLM—to form the final executable function plan (See example plans\_dict.json in GitHub Repository). Together, these output items transform a natural-language filtering rule into a modular, interpretable, and directly executable plan.

\begin{tcolorbox}[
    colback=blue!5,
    colframe=blue!75!black,
    title=\textbf{Prompt for Rule-to-Function Planning (Part 1)},
    fonttitle=\bfseries,
    sharp corners,
    boxrule=0.8pt
]

\textbf{System prompt:} \\
You are a clinical trial filtering planner.

You will be given a single inclusion/exclusion rule for selecting clinical trials (e.g., for landscape analysis).

Your task is to produce a \textbf{structured execution plan} that enables a function writer to:
\begin{itemize}[leftmargin=*, itemsep=2pt]
    \item Extract relevant information from trial metadata using an LLM parser
    \item Define a comparison check against a target value
    \item Specify how multiple checks are combined using logical operators
\end{itemize}

\textbf{Output format:} \\
Return a JSON object with the following fields:
\begin{itemize}[leftmargin=*, itemsep=2pt]
    \item \texttt{filter\_name}: snake\_case name of the filter function
    \item \texttt{logical\_operator}: one of \texttt{["default", "sequential"]}. Specifies how to combine the results of multiple conditions.
    \item \texttt{conditions}: a list of condition dictionaries, each containing:
    \begin{itemize}[leftmargin=1.5em, itemsep=1pt]
        \item \texttt{fields\_to\_attend} : list of trial metadata fields (e.g., \texttt{["Title", "Summary", "Eligibility"]}) relevant *only* for this specific condition.
        \item \texttt{llm\_instruction}: strict instruction for the parser, tailored to extract information needed *only* for this specific condition. Must return either 'Yes'/'No', a number, or a specific string/phrase.
        \item \texttt{comparison}: one of \texttt{["greater\_than", "less\_than", "equal\_to", "not\_equal", "presence\_match", "in\_list"]}
        \item \texttt{target\_value}: the fixed comparison value (must match expected parser output). Can be string, number, or boolean. For "in\_list" comparison, target\_value can be omitted.
        \item \texttt{membership\_list\_name} (optional): For "in\_list" comparison, specify the predefined list name to use for checking membership.

    \end{itemize}
\end{itemize}

\end{tcolorbox}

\begin{tcolorbox}[
    colback=blue!5,
    colframe=blue!75!black,
    title=\textbf{Prompt for Rule-to-Function Planning (Part 2)},
    fonttitle=\bfseries,
    sharp corners,
    boxrule=0.8pt
]
If \texttt{logical\_operator = "default"}, only one condition should be present.

\textbf{You may choose from the following trial metadata fields:}
\begin{itemize}[leftmargin=*, itemsep=2pt]
    \item Title:  Short free-text title of the clinical trial.
    \item Summary: A short paragraph summarizing the study goal, design, and population.
    \item Eligibility: Full free-text inclusion and exclusion criteria. May contain structured bullet points.
    \item Conditions: List of disease conditions the study targets (e.g., \texttt{["Non Small Cell Lung Cancer"]}).
    \item Interventions: A list of dictionaries describing interventions used in the study. Each includes "type" (e.g., DRUG, BIOLOGICAL) and "name" (e.g., "Pembrolizumab").
    \item Study Type: Categorical string such as "INTERVENTIONAL", "OBSERVATIONAL", etc.
    \item Allocation: Describes whether the trial is "RANDOMIZED", "NON-RANDOMIZED", or blank.
    \item Phase: List of strings (e.g., 
    \texttt{["PHASE2", "PHASE3"]}) indicating which trial phases apply.
    \item Primary Outcome: Study primary outcome with measure, description, and time frame.
    \item Secondary Outcome: Study secondary outcome with measure, description, and time frame.
    \item Adverse Event: Adverse event reporting description.
    \item Publications: List of publications related to the trial study. Could potentially be helpful as complimentary information to Summary, Conditions, and Eligibility.
    \item Enrollment: Number of enrolled patients.
\end{itemize}

\textbf{Reference Drug lists:}
\{This input is dynamic from Drug Database via the dedicated Retrieval Agent.\}
\end{tcolorbox}

\begin{tcolorbox}[
    colback=blue!5,
    colframe=blue!75!black,
    title=\textbf{Prompt for Rule-to-Function Planning (Part 3)},
    fonttitle=\bfseries,
    sharp corners,
    boxrule=0.8pt
]
\begin{itemize}
    \item For any condition using "comparison": "in\_list", the llm\_instruction must instruct the parser to extract and return only the drug or intervention names under investigation (not Yes/No), as a Python list (preferred, e.g., \texttt{["trastuzumab", "cisplatin"]}) or as a single comma-separated string. Do NOT instruct the parser to answer whether any are FDA-approved, or to perform the membership check itself. The function writer will perform the membership check outside the LLM using the extracted names and the predefined list. 

    \item If a rule requires a drug, intervention, or other entity to be a member of a predefined list, set "comparison" to "in\_list", and add "membership\_list\_name" with the correct list variable name. The function writer will then perform the membership check outside the LLM, after extraction.
\end{itemize}

\textbf{Do NOT attend to the following fields (already handled during prefiltering):}
\begin{itemize}[leftmargin=*, itemsep=2pt]
    \item NCTId, Status, Completion Date, Has Publication
\end{itemize}

\textbf{Your instruction for each condition must enforce clean responses:}
\begin{itemize}
    \item If the required comparison for that condition is numeric (e.g., greater\_than, less\_than), the instruction should ask for a number: "Return a number only. Do not include units or explanations."
    \item If the required comparison for that condition is binary ('Yes'/'No' extraction) or requires matching a specific string/phrase, the instruction should reflect that: "Return 'Yes' or 'No' only. Do not explain your answer." or "Return the exact phrase that indicates presence, or 'None' if not found. Do not explain your answer."
\end{itemize}
\end{tcolorbox}

\begin{tcolorbox}[
    colback=blue!5,
    colframe=blue!75!black,
    title=\textbf{Prompt for Rule-to-Function Planning (Part 4)},
    fonttitle=\bfseries,
    sharp corners,
    boxrule=0.8pt
]
\textbf{Use the following domain expertise as guidance:}
\begin{itemize}[leftmargin=*, itemsep=2pt]
    \item Decompose complex rules into a list of simpler conditions. Each condition should correspond to a single LLM parser call and a single comparison check.
    \item Choose the logical\_operator:
    \begin{itemize}
        \item 'default': the default one, only one condition to be evaluated.
        \item 'sequential': Conditions are evaluated in order. If a condition evaluates to False, the overall result is True immediately, and subsequent conditions are not checked. This is useful for efficiency when some conditions are much stricter than others (e.g., check phase first, then enrollment within that phase).
      Sequential should always be used for exclusion, therefore, the 'llm\_instruction' should inquire the desired result for subsequential condition. E.g., if the exclusion is great than a threshold, the instruction should ask if it is less than that.
      Example: Exclude phase 4 trials that have enrollment greater than 30. The first instruction is to evaluate whether or not it is phase 4, the second condition should evaluate whether it is *less than* 30.
      Example 2: Exclude phase 2 trials that have enrollments less than 5000. The first instruction is to evaluate whether or not it is phase 2, the second condition should evaluate whether it is *greater than* 5000.
      \item Default to 'and' if there is only one condition to be checked.
    \end{itemize}
    \item Ensure the fields\_to\_attend and llm\_instruction for each condition are minimal and relevant only to that specific condition's required extraction. This allows the function generator to call llm\_parser specifically for the fields and information needed for each step, especially crucial for 'sequential'.
    \item Unless absolutely necessary (usually it is sequential rules), *prioritize single condition checking over multiple conditions (i.e., length(conditions)=1)*.
    \item Unless strictly specified by the rule, do not do literal matching, that is, *do not use quotation marks* for exact matching. Prioritize semantic matching.
    item When evaluating whether a trial reports safety outcomes, prioritize the adverse events section and use primary and secondary outcomes as complementary context. Trials often track safety endpoints through safety populations and AE summaries, even if not labeled as formal outcome measures.
    \item Unless you are absolutely certain or the rule is explicit (such as number of enrollment or phase), check more than one field for each condition as cross reference.
\end{itemize}
\end{tcolorbox}

\subsection{Agentic Drug List Retrieval}
To ensure the fidelity of FDA-approved drug lists used for trial filtering, we designed an agentic retrieval pipeline to collect the required information from validated public sources. The pipeline operates in two stages: local library lookup and web-based retrieval, with intelligent routing between them based on match quality.

The pipeline first examines the trial selection context to identify the disease name (e.g., from fields such as condition or indication). A retrieval agent powered by an LLM then checks a local drug library file that stores previously collected FDA-approved drug lists organized by disease. If a strong match is already available, the pipeline reuses that local list directly. This approach keeps the process fast and ensures consistency across runs, avoiding redundant web queries for previously encountered diseases.

If no strong match exists in the local library, the retrieval agent routes the task to a web searching agent capable of conducting elementary web parsing actions (search, view results, download/parse contents, etc.) to collect the drug list from Drugs.com. The pipeline begins with the Drugs.com search page, enters the disease query, and reviews the returned condition links. It then selects the most relevant condition medication page using fixed matching rules that compare disease terms in the query against words found in the parsed page content.

From the medication table listed on the matched page, the parser extracts medication names row by row. It retains names listed as approved on the page and removes rows marked as off-label. For each medication, it collects the display name, generic name, and brand names, then cleans spacing and removes duplicates. The result is stored as a disease-specific drug list in the local library, ensuring that future runs can reuse it without repeating web collection. This caching mechanism improves both efficiency and consistency across multiple queries for the same disease.

\subsection{Deterministic Execution with LLM-Assisted Parsing}
The execution module instantiates the structured function plan as executable code and applies it to the target data sources. This module is designed to maximize determinism while leveraging the LLM's parsing capabilities in a tightly controlled manner. The architecture separates concerns: the LLM handles only unstructured-to-structured data conversion, while all logical operations, comparisons, and analytic computations are performed deterministically outside the LLM.

For each rule, a function generator produces a concise, self-contained Python function by binding the specified parameters and operators from the plan. These functions are executed via Python's \texttt{exec}, guaranteeing consistent behavior across runs. The generated functions follow a standardized template that first queries only the relevant metadata fields, then invokes the LLM parser with format-constrained instructions, and finally applies deterministic comparison operations to evaluate the condition.

Data parsing is performed by the LLM using the minimal, format-constrained prompts generated during planning. The LLM's role is strictly restricted to converting unstructured text into predefined type, such as booleans (Yes/No), numeric values, or controlled vocabularies, with explicit output format constraints enforced in the instruction. Parsed outputs are validated against expected types and formats; if a field fails validation, the system triggers deterministic fallbacks (such as rule-based string matching or default values) rather than open-ended re-reasoning. This design minimizes the risk of hallucination and ensures that LLM parsing errors do not propagate to downstream logic.

Execution then proceeds through deterministic operations tailored to the analytic task. For landscape-level retrieval, the system applies the compiled filtering functions sequentially to each trial in the dataset and produces structured summaries of eligible trials, including counts, metadata aggregations, and optional visualizations. For eligibility-weighted meta-analysis, the system constructs eligibility tables for selected studies, computes similarity-based weights in the specified feature space using predefined distance metrics, and incorporates these weights into meta-analytic estimates using established statistical routines (e.g., weighted random-effects models). All statistical computations are performed deterministically using standard libraries, ensuring reproducibility and transparency.

\section{Eligibility-Aware Meta-Analysis Implementation Details}
\subsection{Penalty Score Calculation}
To quantify eligibility alignment across trials, we developed a structured framework that converts free-text eligibility criteria into machine-comparable representations and evaluates mismatches using clinically informed penalty scores. The framework consists of three stages: criteria structuring, rule generation with severity assignment, and deterministic penalty computation.

Free-text eligibility criteria were converted into a structured representation using a constrained LLM parsing pipeline. Each eligibility section was decomposed into sentence-level units and mapped into a standardized schema. Specifically, each criterion was represented as a structured tuple:
\[
\texttt{(Type, Entity, Attribute, Value, Condition, Sentence)},
\]
where \textit{Type} denotes inclusion or exclusion; \textit{Entity} corresponds to a clinical category (e.g., demographic characteristics, biomarkers, prior treatment); \textit{Attribute} specifies the variable of interest; \textit{Value} encodes the constraint; \textit{Condition} captures logical qualifiers; and \textit{Sentence} retains the original text for traceability.

Given a target trial and a set of comparator trials, mismatch between trials was quantified using a rule-based penalty framework guided by LLM-derived clinical logic. The LLM first identified clinically meaningful rules for checking discrepancies, producing a set $R$ of comparison rules. Each rule $r \in R$ specifies which clinical entity and attribute to compare, what type of discrepancy to detect (e.g., stricter inclusion, conflicting exclusion, missing constraint), and how severe that mismatch is for the target clinical question. 

Each rule is associated with a severity score $s_r$ ranging from 0 (no discrepancy) to 1.0 (major population incompatibility). Severity scores were assigned by the LLM based on its evaluation of the clinical importance of each mismatch, considering factors such as the type of eligibility constraint, its impact on patient population selection, and its potential influence on trial outcomes. For example, a mismatch in age range might receive a lower severity score ($s_r = 0.2$) if the ranges largely overlap, whereas a conflict in required biomarker status (e.g., HER2-positive vs. HER2-negative) would receive a high severity score ($s_r = 0.9$) due to fundamental differences in target populations.

Each rule was applied to the structured criteria by extracting relevant values for the specified clinical entity and comparing them with the target trial using predefined operators (e.g., $\geq$, $=$, $\in$) via the same function-planning and deterministic-execution approach described in Supplement Section A. A discrepancy was recorded when the comparison condition was violated, triggering assignment of the rule-specific severity score. The penalty for rule $r$ applied to a given trial is defined as:
\[
\text{Penalty}_{r}(trial) =
\begin{cases}
s_r, & \text{if mismatch is present} \\
0, & \text{otherwise},
\end{cases}.
\]

This procedure enabled consistent and reproducible identification of clinically meaningful differences across trials. The overall mismatch score for each trial was computed as the sum of penalties across all rules: 
\[
\text{Total Penalty}(trial) = \sum_{r=1}^{R} \text{Penalty}_{r}(trial).
\]
Higher total penalty scores indicate greater divergence from the target trial's eligibility criteria, and these scores are subsequently transformed into similarity weights for meta-analytic pooling (see Section B.2).

\subsection{Transforming Penalties into Meta-Analytic Weights}

To incorporate eligibility alignment into meta-analytic pooling, penalty scores must be transformed into normalized weights that quantify the relative contribution of each study. This transformation should satisfy several desirable properties: (1) studies with perfect alignment (zero penalty) receive maximum weight; (2) studies with severe mismatches retain non-zero influence to avoid complete exclusion; (3) the decay from maximum to minimum weight is smooth and tunable; and (4) weights sum to one for direct integration into standard meta-analytic estimators.

Let $p_i \ge 0$ denote the total penalty score for study $i$ (as computed in Section B.1), and let $k$ denote the total number of comparator studies. We define $P_{\max}$ as the maximum observed penalty across all comparator studies in the current analysis. This represents the worst-case eligibility mismatch present in the available evidence base. To shape the penalty-to-weight mapping, we introduce two tuning constants: a baseline floor $B \in (0,100)$, which guarantees that every study retains nonzero influence even when maximally misaligned, and a steepness parameter $\gamma>0$, which controls how quickly influence decays as the penalty grows. Higher values of $\gamma$ produce steeper decay, amplifying the distinction between well-aligned and poorly-aligned studies.

Because raw penalties depend on the rule scale and the number of rules generated, we first define a normalized compatibility score that maps penalties to the unit interval. The normalized score is given by:
\[
f_i \;=\; \frac{e^{-\gamma p_i} - e^{-\gamma P_{\max}}}{\,1 - e^{-\gamma P_{\max}}\,}, \qquad 0 \le f_i \le 1.
\]
This exponential transformation ensures that $f_i = 1$ when $p_i = 0$ (perfect alignment) and $f_i = 0$ when $p_i = P_{\max}$ (worst observed mismatch). The exponential decay provides a smooth, differentiable mapping that is sensitive to penalty differences while remaining robust to outliers.

Next, we place this normalized score on a percentage-like range with a guaranteed baseline floor:
\[
S_i \;=\; B + (100 - B)\,f_i. 
\]
Under this transformation, a perfectly matched study ($p_i = 0$) receives an influence score of $S_i = 100$, while the worst-case study ($p_i = P_{\max}$) retains a minimum influence of $S_i = B$. The baseline parameter $B$ prevents complete down-weighting of any study, acknowledging that even misaligned trials may contribute valuable information, particularly when evidence is sparse.

Finally, eligibility-based weights are obtained by normalizing the influence scores across all $k$ studies: 
\[
w_i \;=\; \frac{S_i}{\sum_{j=1}^k S_j}\,.
\]

These weights satisfy $\sum_{i=1}^k w_i = 1$ and can be directly incorporated into standard meta-analytic estimators (e.g., weighted means, random-effects models) to produce eligibility-adjusted treatment effect estimates. Studies with lower penalty scores receive proportionally higher weights, ensuring that the pooled estimate is more heavily influenced by trials with eligibility criteria closely aligned to the target population.

\subsection{Weighted Mantel-Haenszel Estimator}
To incorporate eligibility-based weights into meta-analytic pooling, we extend the classical Mantel-Haenszel (MH) estimator for risk ratios by integrating the normalized weights $w_i$ derived in Section B.2. This section presents the eligibility-weighted Mantel-Haenszel (EW-MH) estimator, establishes its theoretical properties, and provides variance estimation for inference.

For each trial $i=1,\dots,k$, we observe a $2\times 2$ contingency table:
\[
\begin{array}{c|cc|c}
      & \text{Event} & \text{No event} & \text{Total}\\ \hline
\text{Treatment} & a_i & b_i & n_{1i}\\
\text{Control}   & c_i & d_i & n_{0i}\\ \hline
\text{Total}     & m_i & n_i-m_i & n_i:=n_{1i}+n_{0i}
\end{array}
\]
where $a_i$ and $c_i$ denote the number of events in the treatment and control arms, respectively, and $n_{1i}$ and $n_{0i}$ are the corresponding sample sizes.

We assume independent sampling within each trial:
\[
a_i\sim\mathrm{Bin}(n_{1i},p_{1i}), \qquad
c_i\sim\mathrm{Bin}(n_{0i},p_{0i}),
\]
and a common \emph{true} risk ratio across trials (fixed-effect assumption): 
\[
p_{1i}=\theta\,p_{0i}, \qquad \theta>0 \text{ (constant across }i).
\]
We further assume that trial sizes grow at comparable rates, ensuring that no single trial dominates asymptotic behavior. Let $w_i>0$ denote the non-stochastic eligibility weight for trial $i$ as computed in Section B.2. Importantly, $w_i$ is based solely on eligibility criteria alignment and is therefore independent of the observed outcomes $a_i$ and $c_i$.

The EW-MH estimator is defined as:
\begin{equation*}
\widehat\theta_{\text{EW-MH}}
   =\frac{A_w}{C_w}, \qquad
   A_w=\sum_{i=1}^k w_i\frac{a_i n_{0i}}{n_i},\;
   C_w=\sum_{i=1}^k w_i\frac{c_i n_{1i}}{n_i}.
\end{equation*}
Setting $w_i\equiv1/k$ (equal weights) recovers the classical MH estimator $\widehat\theta_{\text{MH}}$. The EW-MH estimator up-weights trials with eligibility criteria closely aligned to the target population, while down-weighting trials with substantial mismatches, thereby producing a pooled estimate that is more representative of the target clinical context.

We now establish the large-sample unbiasedness and variance of the EW-MH estimator.

\begin{proposition}[First-order unbiasedness]\label{prop:unbiased}
Under the fixed-effect model described above,
\[
\mathbb{E}\!\bigl[\widehat\theta_{\emph{EW-MH}}\bigr]
   \;=\;
   \theta + O\!\bigl(n_{\min}^{-1}\bigr),
\]
where $n_{\min}=\min_i n_i$.  Hence
$\widehat\theta_{\emph{EW-MH}}\xrightarrow{p}\theta$ as all $n_i\to\infty$.
\end{proposition}

\begin{proof}
We take expectations of the numerator and denominator of the estimator separately.  Because $a_i$ and $c_i$ are binomial random variables,
\[
\mathbb{E}[a_i]=n_{1i}p_{1i}=n_{1i}\theta p_{0i},\qquad
\mathbb{E}[c_i]=n_{0i}p_{0i}.
\]
Hence
\[
\begin{aligned}
\mathbb{E}[A_w]
   &=\sum_{i} w_i\frac{n_{0i}}{n_i}\,
                n_{1i}\theta p_{0i}
     =\theta
       \sum_{i} w_i\frac{n_{0i}n_{1i}}{n_i}\,p_{0i},\\[4pt]
\mathbb{E}[C_w]
   &=\sum_{i} w_i\frac{n_{1i}}{n_i}\,
                n_{0i}p_{0i}
     =
       \sum_{i} w_i\frac{n_{0i}n_{1i}}{n_i}\,p_{0i}.
\end{aligned}
\]
Therefore $\mathbb{E}[A_w]=\theta\,\mathbb{E}[C_w]$.

Let $X:=A_w$ and $Y:=C_w$. A second-order Taylor expansion of $X/Y$ at $(\mathbb{E}[X],\mathbb{E}[Y])$, along with $\operatorname{Cov}(A_w,C_w)=0$ by trial independence, gives
\[
\mathbb{E}\!\Bigl[\frac{X}{Y}\Bigr]
   =\frac{\mathbb{E}[X]}{\mathbb{E}[Y]}
    +O\!\left(
         \frac{\operatorname{Var}(Y)}{\mathbb{E}[Y]^2}
       \right)
   =\theta + O\!\bigl(n_{\min}^{-1}\bigr),
\]
since $\operatorname{Var}(Y)=O(n_{\max})$ while $\mathbb{E}[Y]=\Omega(n_{\min})$. Hence the bias diminishes as
$1/n_{\min}$ and the estimator is consistent.
\end{proof}
\begin{proposition}[Large-sample variance] \label{prop:variance}
Let $\hat\ell=\ln\widehat\theta_{\emph{EW-MH}}$ be the
log estimate from the eligibility–weighted MH estimator. Then,
\[
\operatorname{Var}(\hat\ell)
   \;=\;
   \frac{\displaystyle
         \sum_{i=1}^k
         w_i^{2}\,
         \dfrac{(a_i+d_i)\,b_i c_i}{n_i^{2}}
       }
       {2\,A_w\,C_w}
   \;+\;O_p\!\bigl(n_{\min}^{-1}\bigr).
\] 
If $w_i\equiv1/k$ the expression reduces to the classical
Robins–Breslow–Greenland (RBG) variance \citep{robins1986general} for
$\ln\widehat\theta_{\text{MH}}$.
\end{proposition}

\begin{proof}

For trial $i$, define
\[
T_i=\frac{a_i n_{0i}}{n_i},\qquad
U_i=\frac{c_i n_{1i}}{n_i}.
\]
Then
\[
A_w=\sum_{i=1}^{k} w_i T_i,\qquad
C_w=\sum_{i=1}^{k} w_i U_i,
\]
and 
\[
\operatorname{Var}(A_w)=\sum_{i} w_i^{2}\operatorname{Var}(T_i),\qquad
\operatorname{Var}(C_w)=\sum_{i} w_i^{2}\operatorname{Var}(U_i).
\]
Applying the multivariate delta method to
$\hat\ell=\ln A_w-\ln C_w$ gives
\begin{equation}\label{eq:delta}
\operatorname{Var}(\hat\ell)
   \;=\;
   \frac{\operatorname{Var}(A_w)}{A_w^{2}}
  +\frac{\operatorname{Var}(C_w)}{C_w^{2}}
  -\frac{2\,\operatorname{Cov}(A_w,C_w)}{A_wC_w}.
\end{equation}

With $\operatorname{Var}(a_i)=n_{1i}p_{1i}(1-p_{1i})$ and $\operatorname{Var}(c_i)=n_{0i}p_{0i}(1-p_{0i})$,
\[
\operatorname{Var}(T_i)=\Bigl(\frac{n_{0i}}{n_i}\Bigr)^{2} n_{1i}p_{1i}(1-p_{1i}),\qquad
\operatorname{Var}(U_i)=\Bigl(\frac{n_{1i}}{n_i}\Bigr)^{2} n_{0i}p_{0i}(1-p_{0i}).
\]
Note that
\[
\frac{\operatorname{Var}(T_i)}{A_w^2}+\frac{\operatorname{Var}(U_i)}{C_w^2}
=\frac{1}{A_w C_w}\!\left\{\operatorname{Var}(T_i)\frac{C_w}{A_w}+\operatorname{Var}(U_i)\frac{A_w}{C_w}\right\}.
\]
By standard binomial plug-in approximations ($a_i/n_{1i}-p_{1i}=O_p(n_{1i}^{-1/2})$, etc.) and
$A_w/C_w=\theta\{1+O_p(n_{\min}^{-1/2})\}$, summing over $i$ and using $p_{1i}=\theta p_{0i}$ yields
\[
\sum_{i=1}^k w_i^2\!\left(\frac{\operatorname{Var}(T_i)}{A_w^2}+\frac{\operatorname{Var}(U_i)}{C_w^2}\right)
=\frac{\displaystyle\sum_{i=1}^k w_i^2\,\dfrac{(a_i+d_i)\,b_i c_i}{n_i^{2}}}{2\,A_w C_w}\,
\{1+O_p(n_{\min}^{-1/2})\}.
\]
Substituting this into \eqref{eq:delta} and using $\operatorname{Cov}(A_w,C_w)=0$ gives
\[
\operatorname{Var}(\hat\ell)
=\frac{\displaystyle\sum_{i=1}^k w_i^{2}\,\dfrac{(a_i+d_i)\,b_i c_i}{n_i^{2}}}{2\,A_w C_w}
\;+\;O_p(n_{\min}^{-1}).
\]
 Setting $w_i\equiv1/k$ recovers the classical RBG variance with
$A=\sum_i a_i n_{0i}/n_i$ and $C=\sum_i c_i n_{1i}/n_i$.
\end{proof}

\paragraph{Confidence Interval Construction.}

It follows from Proposition~\ref{prop:variance} that an asymptotic 95\% confidence interval for $\theta$ is given by:
\[
\text{CI}_{0.95}
  =\exp\!\Bigl(
          \hat\ell
          \pm1.96\sqrt{\operatorname{Var}(\hat\ell)}
        \Bigr),
\]
where $\widehat{\operatorname{Var}}(\hat\ell)$ is obtained by plugging observed counts into the variance formula from Proposition~\ref{prop:variance}.

\paragraph{Scale invariance.}
Because the factor $w_i$ appears multiplicatively in both $A_w$ and
$C_w$, any common multiplier cancels in the ratio $\widehat\theta_{\text{EW-MH}}=A_w/C_w$. Therefore, rescaling the eligibility weights by a constant factor leaves both $\widehat\theta_{\text{EW-MH}}$ and its variance unchanged. This property ensures that the choice of weight normalization (e.g., summing to 1 versus summing to $k$) does not affect inference.

\section{Agentic Trial Filtering for Gastric Cancer Landscape Analysis}
This section provides detailed implementation specifics for the gastric cancer landscape analysis example presented in Section 3.1.  
\begin{tcolorbox}[colback=white,colframe=blue,title=Input Query (lightly copyedited for readability)]
This study aims to identify and evaluate clinical trials on gastric cancer or gastroesophageal junction cancer. Trials must investigate targeted therapies or immunotherapies. Trials should report survival outcomes such as progression-free survival (PFS) or overall survival (OS), and enroll biomarker-stratified populations (including but not limited to HER2-positive, MSI-H, and PD-L1-positive populations). Exclude Phase III trials with small enrollment, for example, fewer than 100 patients. Only include trials investigating FDA-approved drugs.
\end{tcolorbox}

EligMeta generated six inclusion and exclusion rules to operationalize the requirements from the query:

\begin{enumerate}[label=(\roman*)]
\item Include trials that study gastric cancer or gastroesophageal junction cancer
\item Include trials that investigate targeted therapies or immunotherapies
\item Include trials that report survival outcomes such as progression-free survival (PFS) or overall survival (OS)
\item Include trials that enroll biomarker-stratified populations, including but not limited to HER2-positive, MSI-H, or PD-L1-positive
\item Exclude Phase III trials with fewer than 100 enrolled patients
\item Include only trials where the drugs under investigation are FDA-approved
\end{enumerate}

The generated rule set is surfaced for review to allow adjustments by domain experts before any filtering occurs. This human-in-the-loop step ensures that the system's interpretation aligns with the analyst's intent and domain knowledge.

Each rule is then converted into a structured function plan following the methodology described in Supplementary Section A.2. The plan specifies which metadata fields to query, what information the LLM should extract, how the extracted values should be compared, and what logical operators should govern multi-condition evaluation. The system also prepares minimal parsing prompts with explicit type and format constraints (e.g., number, short string, or boolean) to ensure deterministic and reproducible LLM outputs.

As illustrative examples, the plans for rules (v) and (vi) are shown below:

\begin{verbatim}
{
  "filter_name": "exclude_phase_iii_fewer_than_100_enrollment",
  "logical_operator": "sequential",
  "conditions": [
    {
      "fields_to_attend": ["Phase"],
      "llm_instruction": "Check if the trial is in Phase III.
      Return `Yes' if it is, otherwise return `No'. Do not explain your answer.",
      "comparison": "equal_to",
      "target_value": "Yes"
    },
    {
      "fields_to_attend": ["Enrollment"],
      "llm_instruction": "Extract the number of enrolled patients. 
      Return a number only. Do not include units or explanations.",
      "comparison": "greater_than",
      "target_value": 100
    }
  ]
}
\end{verbatim}

\begin{verbatim}
{
  "filter_name": "fda_approved_drugs_only",
  "logical_operator": "default",
  "conditions": [
    {
      "fields_to_attend": ["Interventions"],
      "llm_instruction": "Extract and return the names of drugs under 
      investigation as a Python list. Do not include any other information.",
      "comparison": "in_list",
      "membership_list_name": "FDA_approved_drugs_gastric"
    }
  ]
}
\end{verbatim}

The first plan implements a sequential evaluation strategy: the generated filtering function first instructs the LLM to identify Phase III status from the trial's phase metadata. If the trial is Phase III, the function then extracts the enrollment count and applies a deterministic comparison ($> 100$) to determine retention. The \texttt{sequential} logical operator enables early termination—non-Phase III trials skip the enrollment check entirely, improving computational efficiency.

The second plan demonstrates the separation of LLM parsing from external knowledge validation. The function extracts the list of drugs under investigation from the trial's intervention metadata using the LLM, then performs a deterministic membership check against a curated list of FDA-approved drugs for gastric cancer (retrieved via the agentic drug retrieval pipeline described in Supplementary Section A.3). Critically, the LLM is responsible only for extracting drug names from unstructured text—it does not determine FDA approval status. This architectural separation minimizes the risk of hallucination by confining the LLM's role to parsing tasks where ground truth can be validated, while delegating factual knowledge checks to deterministic external sources.

After executing all six filtering functions on the retrieved trial dataset, the system produces a structured summary of eligible trials, including trial counts, metadata aggregations (e.g., distribution of biomarker populations, phase breakdown), and optional visualizations. This output enables rapid landscape assessment and evidence synthesis for the specified clinical question.

\section{Baseline Comparisons: GPT-5.4 and Agentic Coding}
\subsection{Implementation Details}
We conducted two parallel landscape analyses to benchmark against EligMeta, each representing a distinct paradigm of AI-assisted clinical evidence synthesis.

First, we employed GPT-5.4 with the Deep Research feature available through the OpenAI platform, which reflects one of the most accessible AI tools for clinicians. Deep Research extends standard LLM capabilities by enabling multi-step information retrieval, iterative reasoning, and synthesis across heterogeneous web sources. It autonomously decomposes complex clinical queries, performs sequential searches, evaluates evidence relevance, and integrates findings into structured summaries, thereby approximating a human-like literature review workflow while maintaining ease of use through a natural language interface. 

Second, we utilized the OpenAI Codex client to represent a fully agentic coding paradigm with maximal flexibility and automation. Codex enables end-to-end pipeline construction by generating executable scripts that orchestrate data retrieval, parsing, filtering, and structuring steps. Unlike interface-based systems, this approach allows dynamic control over intermediate representations, custom logic implementation, and reproducible execution, effectively simulating a programmable research assistant capable of reconstructing complex analytical workflows.

For both approaches, identical prompts derived from the gastric cancer use case were provided to ensure consistency in task specification. Additionally, we included an explicit formatting instruction requiring structured tabular outputs. This step ensures comparability across methods, as EligMeta natively integrates this output format within its final result-generation module. The additional formatting prompt is presented below:
\begin{tcolorbox}[
    colback=blue!5,
    colframe=blue!75!black,
    title={Landscape Formatting Prompt},
    fonttitle=\bfseries,
    boxrule=0.6pt,
    arc=2pt,
    left=6pt,
    right=6pt,
    top=6pt,
    bottom=6pt
]
Organize the results and return an Excel file with the following columns:
\begin{enumerate}
    \item NCT Number
    \item Study Title
    \item Intervention(s)
    \item Target/Biomarker
    \item Indication/Condition
    \item Study Phase
    \item Enrollment Size
    \item Status
    \item Trial Summary
    \item Endpoints: Overall Survival
    \item Endpoints: Progression-Free Survival
\end{enumerate}
\end{tcolorbox}

\subsection{Results from GPT-5.4 and Codex}
Given the query without further manual intervention, both systems autonomously executed landscape analyses. We present the workflow and outputs for each approach, followed by a comparative assessment of coverage, accuracy, and adherence to predefined eligibility criteria.

Codex generated a reproducible retrieval and filtering pipeline in Python using the ClinicalTrials.gov API as the primary data source.
The workflow first identified completed interventional studies related to gastric cancer or gastroesophageal junction cancer, and then applied sequential eligibility criteria to retain trials evaluating FDA-approved targeted therapies or immunotherapies in biomarker-stratified populations, including HER2-positive, PD-L1-positive, MSI-H/dMMR, and CLDN18.2-positive disease. Studies were excluded if they were phase IV, had missing phase information, lacked survival endpoints of interest such as overall survival or progression-free survival, or, in the case of phase III trials, enrolled fewer than 100 participants. For studies meeting these structural criteria, survival data were extracted from posted ClinicalTrials.gov results when available and supplemented, when necessary, by ClinicalTrials.gov-linked publications. Filtering steps were conducted via wording checks with regular expression package for natural language processing tasks.

Codex identified 28 gastric and gastroesophageal junction biomarker-stratified trials. Its output was heavily weighted toward HER2-positive studies, which accounted for 21 of the 28 entries, and it included a large number of phase 2 and other exploratory studies. Only 7 of the 28 studies were marked as included in NCCN references. Codex also made several notable selection errors. Most importantly, it missed many pivotal trials, which were among the most practice-relevant studies in the current gastric/GEJ biomarker-guided treatment landscape. In addition, Codex accidentally included FLX475 trial (NCT04768686), which has not been approved by the FDA. These error patterns suggest that Codex performed better as a broad fundamental tool for constructing workflow under expertise's supervision than as a model for precise evidence curation.

GPT-5.4 generated a smaller and more selective evidence set, returning 11 trials with a stronger emphasis on late-stage, practice-relevant studies. Eight of the 11 entries were phase 3 trials, and 8 were marked as NCCN-relevant. The GPT-5.4 workbook also included dedicated Sources and Notes sheets, making its evidence trail more transparent and easier to verify. GPT-5.4 also had important limitations. Although it produced a cleaner and more practice-focused trial set, its narrower retrieval strategy reduced overall sensitivity and likely excluded some eligible or supportive studies captured by Codex. This suggests that GPT-5.4 performed better at prioritizing high-impact, guideline-aligned evidence than at generating a fully exhaustive landscape. In addition, the model also consumed significantly higher computation resources with approximately 1 hour running time with the highest reasoning effort on GPT-5.4 model.

Complete outputs generated by both baseline approaches and the pipeline generated by Codex are provided in the GitHub repository.

\section{Eligibility Criteria of Olaparib Clinical Trials}
\label{sec: eli_four}

This section provides the complete eligibility criteria for the four olaparib clinical trials analyzed in the eligibility-weighted meta-analysis presented in Section 3.2 of the main text. These criteria were extracted from ClinicalTrials.gov and serve as the basis for computing penalty scores and eligibility weights as described in Supplementary Material Section B. The full text of inclusion and exclusion criteria is reproduced below to enable transparent assessment of eligibility alignment across trials and to facilitate reproducibility of the meta-analytic weights. 

\subsection*{Golan 2019 (NCT02184195)}
\paragraph{Key Inclusion Criteria}
\begin{itemize}
  \item Histologically or cytologically confirmed pancreas adenocarcinoma receiving initial chemotherapy for metastatic disease and without evidence of disease progression on treatment.
  \item Patients with measurable disease and/or non-measurable or no evidence of disease assessed at baseline by CT (or MRI where CT is contraindicated) will be entered in this study.
  \item Documented mutation in gBRCA1 or gBRCA2 that is predicted to be deleterious or suspected deleterious.
  \item Patients are on treatment with a first line platinum-based (cisplatin, carboplatin or oxaliplatin) regimen for metastatic pancreas cancer, have received a minimum of 16 weeks of continuous platinum treatment, and have no evidence of progression based on investigator's opinion.
  \item Patients who have received platinum as potentially curative treatment for a prior cancer (e.g., ovarian cancer) or as adjuvant/neoadjuvant treatment for pancreas cancer are eligible, provided at least 12 months have elapsed between the last dose of platinum-based treatment and initiation of the platinum-based chemotherapy for metastatic pancreas cancer.
\end{itemize}

\paragraph{Major Exclusion Criteria}
\begin{itemize}
  \item gBRCA1 and/or gBRCA2 mutations that are considered to be non-detrimental (e.g., ``Variants of uncertain clinical significance'' or ``Variant of unknown significance'' or ``Variant, favour polymorphism'' or ``benign polymorphism'' etc).
  \item Progression of tumour between start of first line platinum-based chemotherapy for metastatic pancreas cancer and randomisation.
  \item Cytotoxic chemotherapy or non-hormonal targeted therapy within 28 days of Cycle 1 Day 1 is not permitted.
  \item Exposure to an investigational product within 30 days or 5 half lives (whichever is longer) prior to randomisation.
  \item Any previous treatment with a PARP inhibitor, including Olaparib.
\end{itemize}

\subsection*{Ledermann 2014 (NCT00753545)}

\paragraph{Inclusion Criteria}
\begin{itemize}
  \item Female patients with histologically diagnosed serous ovarian cancer or recurrent serous ovarian cancer.
  \item Patients must have completed at least 2 previous courses of platinum containing therapy; the patient must have been platinum sensitive to the penultimate chemo regimen.
  \item For the last chemotherapy course prior to enrolment on the study, patients must have demonstrated an objective stable maintained response (partial or complete response) and this response needs to be maintained until completion of chemotherapy.
  \item Patients must be treated on the study within 8 wks of completion of their final dose of the platinum containing regimen.
\end{itemize}

\paragraph{Exclusion Criteria}
\begin{itemize}
  \item Previous treatment with PARP inhibitors including AZD2281.
  \item Patients with low grade ovarian carcinoma.
  \item Patients who have had drainage of their ascites during the final 2 cycles of their last chemotherapy regimen prior to enrolment on the study.
  \item Patients receiving any chemotherapy, radiotherapy (except for palliative reasons), within 2 weeks from the last dose prior to study entry (or a longer period depending on the defined characteristics of the agents used).
\end{itemize}

\subsection*{Moore 2018 (NCT01844986)}
\paragraph{Inclusion Criteria}
\begin{itemize}
  \item Female patients with newly diagnosed, histologically confirmed, high risk advanced (FIGO stage III - IV) BRCA mutated high grade serous or high grade endometrioid ovarian cancer, primary peritoneal cancer and / or fallopian - tube cancer who have completed first line platinum based chemotherapy (intravenous or intraperitoneal).
  \item Stage III patients must have had one attempt at optimal debulking surgery (upfront or interval debulking). Stage IV patients must have had either a biopsy and/or upfront or interval debulking surgery.
  \item Documented mutation in BRCA1 or BRCA2 that is predicted to be deleterious or suspected deleterious (known or predicted to be detrimental/lead to loss of function).
  \item Patients who have completed first line platinum (e.g. carboplatin or cisplatin), containing therapy (intravenous or intraperitoneal) prior to randomisation:
  \item Patients must have, in the opinion of the investigator, clinical complete response or partial response and have no clinical evidence of disease progression on the post treatment scan or rising CA-125 level, following completion of this chemotherapy course. Patients with stable disease on the post-treatment scan at completion of first line platinum-containing therapy are not eligible for the study.
  \item Patients must be randomized within 8 weeks of their last dose of chemotherapy.
\end{itemize}

\paragraph{Exclusion Criteria}
\begin{itemize}
  \item BRCA1 and/or BRCA2 mutations that are considered to be non detrimental (e.g. "Variants of uncertain clinical significance" or "Variant of unknown significance" or "Variant, favor polymorphism" or "benign polymorphism" etc).
  \item Patients with early stage disease (FIGO Stage I, IIA, IIB or IIC).
  \item Stable disease or progressive disease on the post-treatment scan or clinical evidence of progression at the end of the patient's first line chemotherapy treatment.
  \item Patients where more than one debulking surgery has been performed before randomisation to the study. (Patients who, at the time of diagnosis, are deemed to be unresectable and undergo only a biopsy or oophorectomy but then go on to receive chemotherapy and interval debulking surgery are eligible).
  \item Patients who have previously been diagnosed and treated for earlier stage ovarian, fallopian tube or primary peritoneal cancer.
  \item Patients who have previously received chemotherapy for any abdominal or pelvic tumour, including treatment for prior diagnosis at an earlier stage for their ovarian, fallopian tube or primary peritoneal cancer. (Patients who have received prior adjuvant chemotherapy for localised breast cancer may be eligible, provided that it was completed more than three years prior to registration, and that the patient remains free of recurrent or metastatic disease).
  \item Patients with synchronous primary endometrial cancer unless both of the following criteria are met: 1) stage <2 2) less than 60 years old at the time of diagnosis of endometrial cancer with stage IA or IB grade 1 or 2, or stage IA grade 3 endometrioid adenocarcinoma OR $\geq$ 60 years old at the time of diagnosis of endometrial cancer with Stage IA grade 1 or 2 endometrioid adenocarcinoma. Patients with serous or clear cell adenocarcinoma or carcinosarcoma of the endometrium are not eligible.
\end{itemize}

\subsection*{Pujade-Lauraine 2017 (NCT01874353)}

\paragraph{Inclusion Criteria}
\begin{itemize}
  \item Patients must be $\geq$ 18 years of age.
  \begin{itemize}
    \item Female patients with histologically diagnosed relapsed high grade serous ovarian cancer (including primary peritoneal and / or fallopian tube cancer) or high grade endometrioid cancer.
    \item Documented mutation in BRCA1 or BRCA2 that is predicted to be deleterious or suspected deleterious (known or predicted to be detrimental/lead to loss of function).
    \item Patients who have received at least 2 previous lines of platinum containing therapy prior to randomisation.
  \end{itemize}
  \item For the penultimate chemotherapy course prior to enrolment on the study:
  \begin{itemize}
    \item Patient defined as platinum sensitive after this treatment; defined as disease progression greater than 6 months after completion of their last dose of platinum chemotherapy.
  \end{itemize}
  \item For the last chemotherapy course immediately prior to randomisation on the study:
  \begin{itemize}
    \item Patients must be, in the opinion of the investigator, in response (partial or complete radiological response), or may have no evidence of disease (if optimal cytoreductive surgery was conducted prior to chemotherapy), and no evidence of a rising CA-125, following completion of this chemotherapy course.
    \item Patient must have received a platinum based chemotherapy regimen (e.g. carboplatin or cisplatin) and have received at least 4 cycles of treatment.
    \item Patients must be randomized within 8 weeks of their last dose of chemotherapy.
    \item Maintenance treatment is allowed at the end of the penultimate platinum regimen, including bevacizumab.
  \end{itemize}
\end{itemize}

\paragraph{Exclusion Criteria}
\begin{itemize}
  \item Involvement in the planning and/or conduct of the study (applies to both AstraZeneca staff and/or staff at the study site).
  \item BRCA 1 and/or BRCA2 mutations that are considered to be non detrimental (e.g., ``Variants of uncertain clinical significance'' or ``Variant of unknown significance'' or ``Variant, favor polymorphism'' or ``benign polymorphism'' etc.)
  \item Patients who have had drainage of their ascites during the final 2 cycles of their last chemotherapy regimen prior to enrolment on the study.
\end{itemize}

\end{document}